\documentclass{acm_proc_article-sp}

\usepackage{graphicx}
\usepackage{balance}
\usepackage{comment}
\newtheorem{Theorem}{Theorem}
\usepackage{subfigure}
\usepackage{algorithmic}
\usepackage{algorithm}
\usepackage{verbatim}
\newtheorem{theorem}{Theorem}

\begin{document}

\title{Distributed Processes and Scalability in Sub-networks of Large-Scale Networks}

\numberofauthors{6} 

\author{
\alignauthor
Abhinav Mishra \\
       \affaddr{Georgia Institute of Technology}\\
       \email{amishra41@gatech.edu}
}

\maketitle
\begin{abstract}
Performance of standard processes over large distributed networks typically scales with the size of the network. For example, in planar topologies where nodes communicate with their natural neighbors, the scaling factor is $O(n)$, where $n$ is the number of nodes. As the size of the network increases, this makes global convergence over the entire network less practical. On the other hand, for several applications, such as load balancing or detection of local events, global convergence may not be necessary or even relevant. We introduce simple distributed iterative processes which limit the scope of the computational task to a smaller subnetwork of the entire network. This is achieved using one additional local parameter which controls the size of the subnetwork. We establish termination and convergence rate of such processes in theory, in experiment, in comparison to the well understood behavior of Markov processes, and for a variety of network topologies and initial conditions.
\end{abstract}

\section{Introduction}

A  network consists of a collection of nodes that interact and form a network.  Typically, the nodes are embedded either on a plane or on a 3-dimensional space, which limits the convergence
of all fundamental algorithms to scaling factors of roughly $O(n)$ \cite{gupta} , where n is the number of nodes. However, for several applications, convergence happen over the entire network, where as a local conversation is sufficient over a much smaller number of nodes and preferably independent of $n$.  The algorithm proposed here has a scaling factor that depends on the size of local neighborhood and independent of the size of network. It is a non-linear process, therefore its study is of interest.

The convergence of fundamental network processes, such as load balancing or reaching consensus,
is typically determined by well characterized structural properties of the underlying network topology. 
In particular, following the theory of Markov chain mixing, this convergence is of the form $N^{c_0} (1-\lambda)^t$,
where $N$ is the number of network nodes, $c_0$ is a constant, $t$ is time or number of steps,
and $\lambda$ is a quantity polynomially related to expansion, conductance or the second eigenvalue 
of the underlying network topology viewed as a graph. 

Ideally, efficient networks should have $\lambda \geq \frac{1}{\log^c N }$,
which implies convergence rate
\begin{eqnarray*}
N^{c_0} (1-\lambda)^t & = & N^{c_0} (1- \frac{1}{\log^c N})^t \\
    & = & N^{c_0} (1- \frac{1}{\log^c N})^{   \frac{t}{\log^c N}\log^c N          }\\
     & \simeq & N^{c_0} \left(   \frac{1}{e}      \right)^{   \frac{t}{\log^c N}\log^c N           }
\end{eqnarray*}
The above quantity converges quickly for $t\geq \log^{c_0 + c} N$. 
This is efficient, in the sense that it is a polynomial in the $\log $ of the size of the state space.

However, when real network topologies are embedded in 2 or 3 dimensional space, 
$\lambda$ is of the order of $1/N^c$, where $c$ is a constant. 
Consequently, convergence becomes 
\begin{eqnarray*}
N^{c_0} (1-\lambda)^t & = & N^{c_0} (1- \frac{1}{N^c})^t \\
    & = & N^{c_0} (1- \frac{1}{N^c})^{   \frac{t}{N^c}N^c           }\\
     & \simeq & N^{c_0} \left(   \frac{1}{e}      \right)^{   \frac{t}{N^c}N^c           }
\end{eqnarray*}
The above implies that one needs a number of steps polynomial in $N$, i.e., the size of the state space, 
before any meaningful level of convergence is reached. 

Suppose, however, that we do not want global convergence over all $N$ nodes of the entire topology. 
Suppose that convergence over a subset $n \subset N$ nodes, where $n << N$ i.e., $n$ is much smaller than $N$. 
Then, we wish that a suitable process can be defined that indeed converges over the $n$ nodes. 
Moreover, we wish that the convergence rate of this process efficient. 
Since $n << N$, a quantity polynomial in $n$ might be acceptable. 
In addition, one might embed a small graph on $n$ nodes that has $\lambda$ much smaller than the corresponding 
$\lambda$ of the entire network, which would result in even faster convergence rates.

We consider a network, where each node has a potential value called {\it charge}. Charge is similar to the probability in Markov chains. The nodes exchange charge
with their neighbors based on an agreed communication protocol. This is an iterative process and it terminates after few iterations. In the end,
we obtain a tightly connected local component. Size of the component is controlled by the node that starts initiating the charge distribution. 
In the earlier forest fire example, the node that senses fire will initiate the exchange of charge with its neighbors and controls the size of the local component. 
Obtaining such component
is very useful in applications such as load balancing, consensus algorithms, data fusion, etc. In the earlier forest fire example, selecting nodes from this 
component is a superior choice as they have a better chance at detecting fire because of the close proximity to the starting node. 

We present a local iterative algorithm (in Section \ref{sec:algo}), where nodes converge to a value . Our algorithm
is similar to algorithms used in load balancing \cite{load,load4,load3,load2}, consensus algorithm \cite{conmain} (see for a detailed survey), and Markov chains \cite{markov3,matrix,markov1,markov2} . However, it is a non-linear process that makes analysis of the algorithm
a challenging.

\subsection{ Our contributions}

We present a fast algorithm that is based on a non-linear iterative method and terminates after a finite number of steps,
unlike the traditional iterative methods that converge asymptotically. We present a proof for termination and give an upper bound
on the number of iterations before the process terminates.


Our formulation syntactically resembles with iterative methods
such as Markov chain \cite{matrix}, consensus algorithm \cite{conmain}, etc. However, unlike the above mentioned algorithms, our method does not explore the whole graph,
but only a small part of it. This is where locality comes in. There is another major difference between our algorithm and other iterative techniques.
Our algorithm is a non-linear process, which makes it difficult to analyze and prove guarantees. To the best of our knowledge, this is the first work on
 non-linear iterative method with  a locality preserving property. We do prove the termination and an upper-bound on the number of iterations.
We later show that the algorithm is in fact a generalization of markov chain. 

As discussed earlier, there are many applications of our algorithm. We allow multiple load in load balancing. We show that our algorithm
gives a solution to problem known as $f-consensus$ and also, any decomposable function can be used in solving consensus problems.
We later show the application in activity and event detection problems, and how information from different  nodes can be combined
to obtain a result. The problem of combining information from multiple is known as {\it data fusion}.

We now give an outline of the paper. In Section 2, we present an outline of the algorithm and succinctly present its application. 
In Section 3, we give definitions of different type of nodes in a graph, and their properties. We discuss the formation of network from 
 nodes and communication protocol between the nodes, i.e., a node can interact with its neighbor only. In Section 4, we
formally present the algorithm, and show an implementation using a fixed-point iteration. In Section \ref{sec:prop}, we prove various properties 
associated with the algorithm that are later used to derive an upper-bound on the number of iterations required for termination (Section 7). We show some generalization and extensions of the algorithm in
In section \ref{sec:gen}. Later, we present various application of our algorithm namely, consensus algorithm (Section \ref{sec:cons}), load balancing (Section \ref{sec:load}). We then present the experiments, related work and conclude.

\begin{figure*}
\centering
\includegraphics[width=1\textwidth]{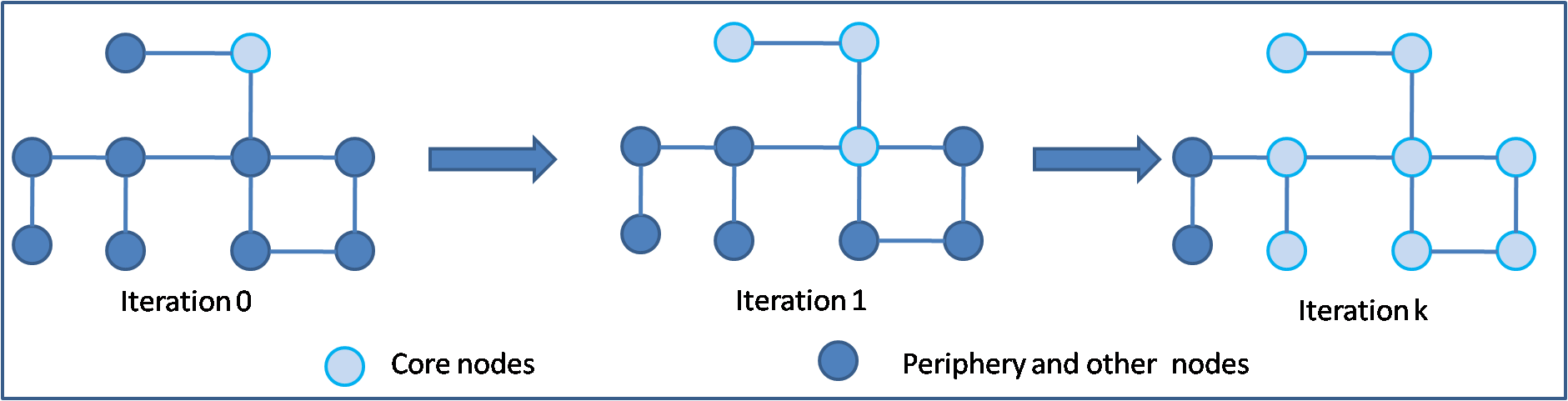}
\caption{Initially, there is one core node and then, in first iteration it distributes the charge to its neighbors increasing the size of the core.
After k iterations, we see a bigger core size.}
\label{corevsperi}
\end{figure*}

\section{Outline of algorithm}
\label{sec:outline}

In a graph, each node holds a charge. A charge is as a potential value of a node. It is a probability distribution in case of Markov chains \cite{matrix}. However, we generalize it by introducing the concept of charge, and do not constraint the charge to sum up to 1 across all the nodes. Therefore, instead of saying that a node has a probability of $0.3$, we say that it has a charge of $0.3$. We formally prove this generalization in Section 7.
Charge can be seen as the amount of flow a node receives in flow problems.
At every iteration, a node gives a
charge to its neighbors and in-turn also receives charge from them. And we repeat this process. After a certain 
number of iterations, this process stabilizes and there is no exchange of charge between nodes. Therefore, we say that the process terminates. 
For simplicity, let assume that if a node has more charge
than a certain threshold $\epsilon$, then it redistributes it to its
neighbors a part of it. 


We later show that at the end of the process, a group of nodes form a {\it core}. This core forms a small
tightly connected component. This set of core nodes is  useful in many applications such as, in load balancing (see Section~\ref{sec:load}), if a node has high load, it may want to redistribute it to the nodes close to it, but not further away. In activity tracker , if a node senses an unusual event, it would notify the nodes in a close proximity to reduce false positives. We discuss these applications in detail in later sections. 

Also, note that the size of the core is essentially a function of $\epsilon$. If $\epsilon$ is small, then we end up with a huge core and like wise
$\epsilon$ is large we have a small core. Depending on an application we may want a large or small core.

\begin{table}[t]
\begin{center}
\begin{tabular}{|c|c|c|}
\hline
Node type & Charge range \\
\hline

Core Nodes & $(\epsilon,1]$ \\
Peripheral nodes & $[0,\epsilon)$ \\
Other nodes & 0\\
\hline
\end{tabular}
\caption{There is four kinds of nodes in the graph. Active nodes still belong to the core, but have an ability to distribute charge to other nodes.
Peripheral nodes are the neighboring nodes of core. }
\label{tab:charge}
\end{center}
\end{table}

\section{System Model and Definitions}
\label{sec:model}

Consider $G(V,E)$ as a connected undirected graph, $|V| \! = \! n$,
$d_i$ is the degree of node $i$ and $d_{\max} \! = \! \max \{ d_i : i \in V \}$, and
$\vec{x^0} = (x_1^0 , \ldots , x_n^0)$ is a non-negative function over the set of nodes: $x_i^0 \! \geq \! 0$. 
We call $x_k^i$ as the charge of node $i$ at iteration $k$. We assume that the charge is non-negative for any node at any iteration.
In Markov Chain, we additionally put a constraint that $\sum_{i \in V} x_i^0 = 1$, so that it becomes a probability distribution. However,
such constraint is not necessary in our algorithm. For simplicity, we assume that there is only one node in the beginning with some positive charge,
and charge for rest of the nodes is zero. We later remove this assumption in Section \ref{sec:multi}.

We consider a network, where a node communicates with its neighbors. For a node to 
communicate with a node 3-hops away, it does it through the neighbors. We now present the model for formation of such network. 
There are many ways a group
of nodes can form a network. For example, if a node $i$ is in a top-k neighbors of node $j$ based on a certain metric, then form 
a directed edge from $j \rightarrow i$. Another method, is a distance based, where all the nodes within a certain
radius are automatically treated as neighbors. Notice that the second method gives a symmetric neighborhood, where if a
node $i$ is a neighbor of $j$, it implies that the node $j$ is also a neighbor of node $i$. However, in first method this is not a case, as 
two nodes may not be in each other's top-$k$ neighborhood at the same time. Therefore, we can either consider a situation,
where a mere presence of a directed edge imply the neighborhood, or we say that if there is a directed edge from both sides, it implies
neighborhood \cite{kn}. We now present the definitions of different type of nodes in the graph. Please refer to Table~\ref{tab:charge} for the
node types and charge.

{\bf Core nodes:} A core is a group of nodes that satisfy the condition $ x_i^* \geq {\epsilon} $, $\forall i \in G.$ Essentially,
all the nodes in the core have a minimum charge of $\epsilon $. As we shall later see that once a node becomes a part of core, i.e., once its
charge is more than $\epsilon $, it continue to be a part of core irrespective of iterations, as its charge does not fall below $\epsilon $. Another interesting thing to note that the core nodes themselves form a connected component (see Section \ref{sec:prop}), we later prove this property.

{\bf Peripheral nodes:} The peripheral nodes is a group of nodes that satisfy the condition $0 \leq x_i^* < \epsilon$, $\forall i \in G.$ 
Please refer to Table~\ref{tab:charge} for the
node types and charge.
These nodes are connected to the nodes in the core, but have a charge less than $\epsilon $. These nodes are in a periphery of core.
Since the charge of a peripheral node is less than $\epsilon$, it is not allowed to redistribute the charge. A peripheral node can receive charge from core nodes. Notice that any node that is connected to a peripheral node, but not to the core will have a zero charge, as it cannot receive the 
charge from peripheral node. Also, a node initially could be a part of peripheral nodes, but can become a core once its charge is more than $\epsilon$. However, a core node cannot become a peripheral node in any iteration.

{\bf Charge conservation of the graph:} At any iteration, the total amount of charge in the graph remains the same.

{\bf Charge conservation at the node:} The difference in charge in a node between two consecutive iterations is the difference between the incoming
and outgoing charge to that node.

So far, we have discussed the properties of core nodes and peripheral nodes. In next section, we precisely define the process of exchange of charge between the nodes. 

\begin{figure}
\centering
\includegraphics[width=0.5\textwidth]{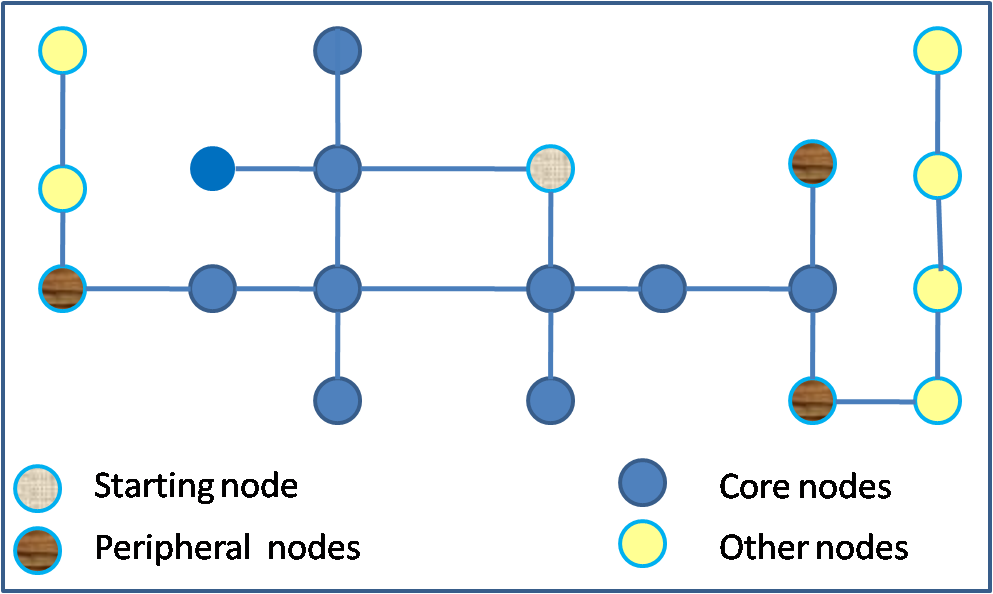}
\caption{Initially, there is one core node and in first iteration it distributes the charge to its neighbors increasing the size of the core.}
\label{corevsperi}
\end{figure}

\section{Algorithm}
\label{sec:algo}

For a node $i$ at iteration $t$, we define a variable $z_i^t$ such that

\begin{displaymath}
z_i^t \,\,= \,\, \left\{ 
\begin{array}{lll}
1 & {\rm if} & x_i^t > \epsilon \\
0 & {\rm if} & x_i^t \leq \epsilon
\end{array}
\right. 
\end{displaymath}

\noindent
For any node with $z_i^t=1$ imply that it is a part of core. It is a necessary and sufficient condition. We use $z_i^t=1$ and active node interchangeably.
 In our algorithm, the nodes with $z_i^t=1$ are only nodes allowed to distribute
the charge to their neighbors. In other words, nodes with a charge more than $\epsilon$ can transfer the charge.

For a node $i$ at iteration $t$, the amount of charge it carry in iteration
$t$ is represented by $x_i^{t}$, and in $t+1$ iteration by $x_i^{t+1}$. Let $d_j$ be the degree of node $j$. The charge $x_i^{t+1}$ is computed as follows:

\begin{equation}
\label{p1}
x_i^{t+1} \,\,= \,\,\left( \epsilon z^t_i+\frac{( x_i^t-\epsilon)}{2}  z_i^t + x_i^t (1-z_i^t) \right) \,\,+ \,\,
\frac{ 1}{2}\sum_{j : \{ i,j \} \in E } \frac{(x_j^t-\epsilon)z_j^t}{d_j} 
\end{equation}

\noindent

The first part of the equation  indicates the amount of charge node $x_i^{t+1}$ retains.
If $z_i^t=0$, then the node retains all the charge it had in the previous iteration, and therefore by the conservation of charge at the node , it
cannot send charge to its neighbors. If $z_i^t=1$, then the node retains $\epsilon$  charge as given by $\epsilon z^t_i+\frac{( x_i^t-\epsilon)}{2}  z_i^t $, and it sends rest of the
charge to its neighbors.

Second part of the equation $\frac{ 1}{2}\sum_{j : \{ i,j \} \in E } \frac{(x_j^t-\epsilon)z_j^t}{d_j}  $ indicates the amount a node receives from its neighbors.
Since only the node, with $z^t_*=1$ are allowed to transfer the charge, therefore we have an indicator variable $z^t_*$ indicating if the node is eligible to send the 
charge. The nodes with $z^t_*=1$ can only distribute the other $ \frac{(x_j^t-\epsilon)}{2d_j} $ fraction to its neighbors. In our process, they distribute
it equally to their $d_j$ neighbors. Here $d_j$ indicates the degree of node $j$. It is not difficult to verify the charge conservation with
the above formulation

Note that above formulation is a non-linear system of equation that terminates 
after finite steps using fixed-point iterations. This non-linearity makes it incredibly 
difficult to analyze, as the standard methods from linear algebra cannot be applied.
 We prove the termination of our
algorithm in Section 7 and also give an upper-bound on the number of iterations required
for a guaranteed termination. The algorithm exhibits a locality property that
the charge is disseminated to the nodes that are closer in proximity. We discuss this locality in detail
with applications in Section~\ref{sec:local}.

\subsection{Properties of the algorithm}
\label{sec:prop}
For a subset of nodes $H \! \subset \! V$, we use $\Gamma_1 (H)$ to denote
the 1-vertex neighborhood of $H$: 
$\Gamma_1(H) = \{ j \! \in \! V \setminus H : 
\{ i^\prime \! \in \! H : \{ i^\prime , j \} \! \in \! E \} \! \neq \! \emptyset \}$.
\begin{theorem}
Let $i_0 \! \in \! V$ be a node such that $x_{i_0}^0 \! = \! 1$ (consequently $x_i^0 \! = \! 0$, 
$\forall \! i \! \neq \! i_0$) and let $n \geq\frac{1} {\epsilon}$. 
Then, there exists core group of nodes $H \! \subset \! V$ with $i_0 \! \in \! H$ 
and there exists $T \! >\! 0$ such that :\\
$~~~~${\bf (i)} ${\epsilon}< x_i^* \leq 1$, $\forall i \in H$. \\
$~~~~${\bf (ii)} The subgraph of $G$ induced by $H$ is connected. \\
$~~~~${\bf (iii)} $x_i^{t+1} = x_i^t$, $\forall i \in V$ and $\forall t \geq t_0$. \\
$~~~~${\bf (iv)} $0 \leq x_i^t < \epsilon$, $\forall i \in \Gamma_1 (H)$. \\
$~~~~${\bf (v)} $x_i^t = 0$, $\forall i \in V\setminus\left( H \cup \Gamma_1(H) \right)$
and $\forall t \geq 0$.\\
$~~~~${\bf (vi)} $|H|\leq 1/\epsilon$. \\
$~~~~${\bf (vii)} $|V|\leq 1/\epsilon$, then algorithm does not terminate. 
\end{theorem} 
{\sc Proof of Theorem 1.} 
We assumed that there is a single starting node with charge 1 ( we relax this constraint in Section \ref{sec:multi}). All the other nodes in the beginning have a zero charge.
{\bf(i)} If a node has a charge more than $\epsilon$, it will retain  $\epsilon$  charge (by Eq. 1). All the other nodes with charge more than $\epsilon$, will continue to hold all of the charge and will not distribute it to their neighbors. {\bf(ii)} A node can only receive a charge
from a core . Therefore there has to be a path of core nodes between a node receiving a charge and $i_0$ (starting node). Also, a node that becomes
a core node continues to be a core node, proving the connectivity. {\bf(iii)} See Section7 {\bf(iv)} Here, $ \Gamma_1(H) $indicates the set
of nodes that are neighbor of core including the core nodes itself. $ H \cup \Gamma_1(H)$ shows the peripheral nodes. Such peripheral nodes can
receive the charge from core since they are the neighbors of core, but their charge has to be less than $\epsilon$, otherwise they become core.
{\bf(v)} All the nodes that neighbors of peripheral nodes, but not a part of core will have zero charge, as the peripheral node can have 
at most $\epsilon$ charge and are not allowed to distribute it to their neighbors.{\bf(vi)} The minimum charge a core node can have is $\epsilon$ and 
charge of a starting node $1$, there by the conservation of charge, there could be at most $1/\epsilon$ core nodes. {\bf(vii)} If the size of the graph is less than $1/\epsilon$, then there is at least one node with charge more than $\epsilon$ and will continue to give charge to
its neighbors, therefore the process cannot terminate.


\section{ Convergence Properties}
In this section we prove the convergence properties. If we prove the convergence on the core nodes, it automatically prove it for the whole graph
as the peripheral and other nodes can not transfer the charge. We show the convergence in steps. Initially,
we prove that any two nodes across an edge in a core converge to $epsilon$ exponentially fast. To prove these properties, we assume that the graph is regular. Later we show with a chain of arguments that nodes in a core at any iteration converge in a finite time.  
Before we proceed further, we prove the following inequality that will be subsequently used to prove the convergence properties.
Let $C^t$ be the set of core nodes at iteration $t$. Then, for a node $i \in C$
\begin{align*}
\label{p1}
|x_i^{t+1} - x_i^{t}| \,\,&= 
\left | \frac{ x_i^t -x_i^{t-1}}{2}  \,\,+ \,\,
\frac{1}{2d} \sum_{j : \{ i,j \} \in E,j \in C } {(x_j^t-x_j^{t-1}} ) \right|\\
&\leq  \frac{ |x_i^t -x_i^{t-1}|}{2}  \,\,+ \,\,
\frac{1}{2d} \sum_{j : \{ i,j \} \in E,j \in C } {|x_j^t-x_j^{t-1}} | \\
&\leq 
\frac{1}{2d} \sum_{j : \{ i,j \} \in E } {|x_j^t-x_j^{t-1}} | + |x_i^t -x_i^{t-1}|\\
\end{align*}

\subsection{Error bound}
Now we bound the error for an edge with in a core. We show that for a pair of nodes of an edge, the sum of the difference of their charge start
to decrease exponentially fast. We prove it using mathematical induction.

\begin{Theorem}
The sum of difference in charge  of two nodes between two consecutive iterations 
$t$ and $t+1$ is bounded by an inverse exponential function of $t$:
\begin{align}
|x^{t+1}(j) - x^t(j)|+|x^{t+1}(i) - x^t(i)| \leq {\left(\frac{1}{2}\right)^{t}}.
\end{align}
\end{Theorem}
\begin{proof}
We prove using mathematical induction.

\textbf{Basis}: We first prove for $t=1$.
\begin{align*}
\sum_{i,j}|x^{2}(i) - x^{1}(i)| \leq
\frac{1}{2d} \sum_{k : \{ i,k \} \in E } {|x_k^1-x_k^{0}} | +
 |x_i^1-x_i^{0}| +\\
\frac{1}{2d} \sum_{m : \{ j.m \} \in E } {|x_j^1-x_j^{0}} | +
 |x_m^1-x_m^{0}| 
\leq \frac{1}{4}
\end{align*}

This confirms the basis. The last equality holds because at most half of the unit charge can be distributed initially. Therefore the sum of absolute value across
all nodes could be at most $1/2$. 

\textbf{Induction step}: We assume the bound to be true for $x^t(i)$, i.e., for the
$t^\text{th}$ iteration. In the $(t+1)^\text{th}$ iteration,

\begin{align*}
\sum_{i,j}|x^{t+1}(i) - x^{t}(i)| \leq
\frac{1}{2d} \sum_{k : \{ i,k \} \in E } {|x_k^t-x_k^{t-1}} | +
 |x_i^t-x_i^{t-1}| +\\
\frac{1}{2d} \sum_{m : \{ j.m \} \in E } {|x_j^t-x_j^{t-1}} | +
 |x_m^t-x_m^{t-1}| \\
\leq \frac{1}{2} \cdot  \left ( \frac{1}{2} \right )^{t} = \left ( \frac{1}{2} \right )^{t+1}    \mbox{[Induction Assumption]}\\
\end{align*}
\end{proof}

It is not difficult to see that we need to perform $O(\log{1}{/\epsilon})$ iterations in order to make $|x(i)-\epsilon|<\delta$. This shows that the algorithm is really fast, as we can make any node come to close $\epsilon$ in few iterations. We use this error bound to establish the convergence. If we set the parameter $\epsilon$ as 0, then our algorithm   becomes a lazy random-walk.

\subsection{Proof of convergence}
Before we give a  proof of convergence. We discuss the high level idea. At time $t$, let $C$ an active core with access charge. Just to remind the excess charge is defined as the sum of all the extra charge across all the nodes that can be transferred to their neighbors. With time, the access charge continue to flow towards periphery, and excess charge  will continue to decrease. As long as the size of the network is more than $1/\epsilon$, a convergence will be reached. It is easy to see that if a network is small, then there will not be a convergence sa there has to be at least one node with
excess charge.

Let $C^i$ be the set of core nodes at iteration $i$. Using the bound proved earlier, we establish that any node $k \in C^i$ will reach convergence.
However, in order to reach convergence, it has to give charge out. Since whole $C^i$, reaches convergence it implies that new nodes from periphery and beyond have received the charge and forming a new core $C^j$ at iteration $j$.  The gap $j-i$ is finite because in $C^i$ converges exponentially fast. We 
can apply the same arguments to $C^j$ and since the maximum size of core is bounded by $1/\epsilon$, there could be at most $1/\epsilon$ new cores.
This proves the convergence.

\section{ Extensions}
\label{sec:gen}

\subsection{Nodes with Different Thresholds}
\label{sec:thres}
In this section, we discuss the scenario where nodes have different thresholds, i.e., there is a separate threshold $\epsilon_i$ for each node $i$.
This scenario is especially useful in applications such as load balancing (see Section \ref{sec:load}), where nodes have different load capacities. For example, a node can support a maximum capacity of say, 5 units, where as a different node can have a capacity of 100 units. We encode such information using the 
parameter $\epsilon$.
We discuss more on application in Section \ref{sec:load}. For each node $i$, we define variable $z_i^t$ as follows:

\begin{displaymath}
z_i^t \,\,= \,\, \left\{ 
\begin{array}{lll}
1 & {\rm if} & x_i^t > \epsilon_i \\
0 & {\rm if} & x_i^t \leq \epsilon_i
\end{array}
\right. 
\end{displaymath}

Note that even when using a different $\epsilon_i$, we still maintain the conservation of charge both at the node level and consequently on the whole
graph. Eq. (1) remains the same.

\subsection{Multiple starting nodes}
\label{sec:multi}
There are situations where we require multiple starting nodes and the net charge across the nodes is more than 1.
For example, in case of a forest fire, multiple sensors detect fire at the same time. Therefore, we have a multiple starting nodes.

Let us assume that there are two nodes $i$ and $j$ that have an initial charge of more than $\epsilon$ and therefore,they can 
dissipate it to their neighbors. Let $H(i)$ and $H(j)$ denote the combined set of core nodes and peripheral nodes formed by the nodes $i$ and $j$
respectively. Now,
there are two scenarios possible {(i)} $H(i) \bigcap H(j) = \phi$ and {(ii)} $H(i) \bigcap H(j) \neq \phi$. In second case 
where two subgraphs induced by node $i$ and $j$ do not intersect, imply that the two starting nodes do not impact each other.

It is very important to include peripheral nodes in sets $H(i)$ and $H(j)$, as it is possible that a node receives less $\epsilon$ charge 
from nodes $i$ and $j$ independently, but has a combined charge of more than $\epsilon$, and the node becomes active. This gives rise to the first scenario where set $H(i)$ and $H(j)$ do have some common nodes. $H(i)$ and $H(j)$ will continue to grow until they intersect. Then some of the nodes in the intersection might become active. If there are no active nodes, then this scenario is similar to the second case where $H(i) \bigcap H(j) \neq \phi$.
Otherwise, two subgraphs merge and form a single core, and algorithm continue to run.

\section{Local Algorithm}
\label{sec:local}
In Eq.(1), we introduced an iterative algorithm, where each node in a graph updates its value(or charge)
after exchanging information with neighbors. After a few iterations, this decentralized process terminates. An interesting property to note is that algorithm runs on a small part of graph and that is solely depended on $\epsilon$. This makes our algorithm very favorable 
compare to other iterative algorithms such as Markov Chains \cite{markov3}, random walk with restarts \cite{tong}, consensus algorithm \cite{conmain}, as they require
the whole graph at their disposal.

To the best of our knowledge, this is the first local iterative algorithm that does run on a small part of a graph (core nodes) and terminates. 
This makes it desirable for the applications where we need to look up in a small neighborhood. For example, if a node detects a signal wants a group of node to track an intruder in a region. It can build a set of core nodes (that are in close proximity) by running the algorithm. This scenario is very useful in applications such as activity tracking or border surveillance \cite{article,prob,target1,target2}. Where a group of sensor nodes is required to track the movement. Here, nodes in core start sensing for a particular
signal instead of all the sensors in the graph.
Another example is event detection. For example,
if a sensor node reports a fire event, it is necessary to validate it multiple nodes to reduce the false-positive rate. This also makes the system fault resilient.

In many applications, we require to compute a function on a set of nodes. The function could be a mean distributed averaging \cite{conmain}, Kalman filter \cite{kalman}, Maximum likelihood estimator \cite{boyd}, linear least-squares estimator \cite{mean}, max/min, or a decision based function. For example, in case of consensus algorithms, the function is usually averaging. In case of forest fire,
sensor nodes could report a binary value indicating if they sense any fire or not. This is an example of decision based function and computing average is a value based
function. The problem of combing data from multiple sensor nodes has been well studied in sensor nodes and is referred as {\it data fusion}. We discuss
data fusion in detail in Section \ref{sec:fusion}. For now, we focus on the function we want to compute in core nodes. The function needs to have a property 
known as decomposability, i.e., it should be computable in a distributed setting over the set of nodes. For example, maximum can be computed, as 
the nodes can update their value to the maximum of its own value and its neighbors, and we iterate this process.

We now give details of our local algorithm. We have two components to our local algorithm: Algorithm \ref{alg:core} and Algorithm \ref{alg:fun}. In first component (Algorithm \ref{alg:core}), we build the core set. This
is achieved by running the iterative algorithm as described in the Eq. (1). In second component (Algorithm \ref{alg:fun}), we compute the function over the set of core nodes. 
In first algorithm, where we form a core set. Vector $\vec{x^0}$ (Line 1) represent the charge of the nodes. The  nodes that have 
charge more than $\epsilon$ belongs to the core. After we run this algorithm, we return the updated vector $\vec{x^t} $ and the core set $C$. In line 4, we iterate the Eq. (1) until we terminate. We can either set the number of termination steps apriori using the upper
bound we proved, or using a termination criteria, i.e., $\sum_i |x^k_i-x^{k+1}_i|=0$. However the later method is more computationally expensive, as
it requires storing the difference between the two iterations and it has to be communicated to other nodes. At the same time, while being 
computationally more expensive, this approach requires fewer iterations compared to the upper bound. In line 6, for every node in core and its neighbors,
we run the iterative algorithm as described in eq (Line 6). The reason to include the neighbors is that the core node with more than $\epsilon$ charge 
will give charge to the neighbors also that are not the part of core yet. In the next step (line 7), we check whether the node is a part of core or not (line 8).
This is achieved by checking if their charge is more than $\epsilon$. With each iteration, we aim to increase the size of core and we stop once we 
cannot increase the size core using the termination criteria or upper bound. At the end of the this algorithm, we return the updated charge values and set
of core nodes (line 12).

Now we move to the second component of the algorithm (Algorithm \ref{alg:fun}) where we compute a decomposable function over a set of nodes. In previous algorithm (Algorithm \ref{alg:core}), we
obtained the core set. Now, we just have to run the decomposable function on the core set. Note that we are not going to discuss the exact implementation 
of such functions as it varies with each function. However, it could be any function that is decomposable, i.e., the function can be computed in 
a distributed setting.

\begin{algorithm}
\caption{Core building algorithm}
\label{alg:core}
\begin{algorithmic}[1]
\STATE {\bf Input:} $\vec{x^0} = (x_1^0 , \ldots , x_n^0)$ 
\STATE Let $C$ be a set of core nodes; 
\STATE {\bf Output:} $\vec{x^t}, C$ 
\REPEAT 
\FOR{node $i$ in $C \cup Neighbors(C)$}
\STATE $x_i^{t+1} \,\,= \,\,\left( \epsilon z^t_i+\frac{( x_i^t-\epsilon)}{2}  z_i^t + x_i^t (1-z_i^t) \right) \,\,+ \,\,
\frac{ 1}{2}\sum_{j : \{ i,j \} \in E } \frac{(x_j^t-\epsilon)z_j^t}{d_j} $ 
\IF{$x(i)>\epsilon$}
\STATE $C=C \cup\{ i\}$
\ENDIF

\ENDFOR
\UNTIL {Termination}
\STATE {\bf return} $\vec{x^t},C$; 
\end{algorithmic}
\end{algorithm}

\begin{algorithm}
\caption{Function computation algorithm}
\label{alg:fun}
\begin{algorithmic}[1]
\STATE {\bf Input:} $\vec{y^0} = (y_1^0 , \ldots , y_n^0)$ 
\STATE Let C be a set of core nodes 
\STATE {\bf Output:} $ \vec{y^t}$ 

\FOR{node $i$ in $C $}
\STATE $Y_i^{t+1}=f(Neighbors(i))$

\ENDFOR
\STATE {\bf return} $\vec{y^t}$; 
\end{algorithmic}
\end{algorithm}


\section{Applications}
\label{sec:app}

\subsection{Consensus Algorithms}
\label{sec:cons}

In a network , consensus signify an agreement between nodes on a certain
quantity of interest. For example, in case of a forest fire, then we expect sensor nodes
to come to an agreement that there is a fire and raise an alarm. If there is just one node
indicating the presence of fire, this might imply a false positive due to malfunctioning of a node. 
A consensus algorithm is a protocol involving information exchange between the sensor nodes to come
to an agreement. These algorithms have been studied in great detail in Markov chains \cite{conmain}, and often used in load
balancing algorithms \cite{muthu}. 

Recently consensus algorithms play a major role in in  networks. In particular, 
many applications require a global clock synchronization,
that is all the nodes of the network need to share
a common notion of time \cite{sync1}. Therefore, clock synchronization has been an active area of research (see \cite{sync2} for a comprehensive survey).
Consensus algorithm have been heavily used in data fusion methods such as using kalman filtering \cite{con1}. We later discuss data fusion 
in detail in Section \ref{sec:fusion}.

Since all the core nodes have the $\epsilon$. The algorithm
algorithm gives a solution to a problem known as $f$-consensus\cite{conmain} where all nodes are asymptotically have the same value.

To the best our knowledge, we have not come across any local consensus/$f$-consensus algorithm that runs on small part of graph. We now present a further extension and show
that using algorithm 2, we can turn any consensus algorithm locally. As consensus is an information exchange between the neighboring nodes and each
node applies a function on the information it exchanged from its neighbors. The function could be averaging operator, graph laplacian or any other operator.
If we plug this function in algorithm 2, we obtain a local version of that consensus algorithm.

\subsection{Load Balancing}
\label{sec:load}

Load balancing distributes workloads across multiple resources, which in our case is network nodes. By balancing the load across the nodes, we improve
the energy utilization and thus extends the lifespan of the  network \cite{load}. In another example, sensor's energy might not be 
sufficient to support long communication and may require hops to forward the information \cite{load2}. This makes the load balancing a well studied
problem in context of routing in networks \cite{load3,load4}.

In our work, we assume that each r node has a finite capacity. This is a reasonable assumption as sensors have a limited battery life
and processing power.
Therefore, there is an upper-limit on the maximum load a node can handle. We allow each node to have a different load. The problem with multiple load 
capacities can easily be modeled 
with our algorithm. Moreover, it is preferred that a load is first given to a close neighbor compared to the node which is far away. 
This preference has two advantages 1) smaller network delay and 2) lower energy consumption. 

Let $\epsilon_i$ represent the maximum load the node $i$ can handle, and let $L_i$ be the current load at node $i$. Therefore,
the amount of extra load the node $i$ can handle is $\epsilon_i- L_i$, otherwise this difference has to be distributed to the other nodes.
We set the initial charge vector (in Algorithm \ref{alg:core}) $\vec{x}= (L_0,\ldots,L_n)$. In section \ref{sec:multi} , we discussed the relaxation to allow multiple thresholds,
that is a different $\epsilon_i$ for node $i$. Note that in case of load balancing this $\epsilon_i$ precisely represent the maximum load capacity of node $i$. 
It is
out of scope of current work to discuss the schemes to identify the maximum load of a  node. However, it could depend on few parameters such as battery life, maximum distance it has to transmit/sense, or the amount of processing.

Algorithm \ref{alg:core} with the current settings precisely solves the load balancing problem.
utilization of core nodes. If a node exceed its load capacity, i.e., $L_i>\epsilon_i$, it will transmit the data, as the variable $z_i=1$ in such case (see section 5.2). The algorithm favors the nodes that are in a close vicinity. This greatly helps in improving energy consumption and
reduce network delays.

\subsection{Data Fusion}
\label{sec:fusion}
There are many applications in event detection that require data fusion \cite{collab,sensordeploy2002}. Another example of decision based fusion
monitoring temperature in a building 
where each sensor senses the temperature and a central unit decides to raise or lower the temperature.
Here, the central unit might not be interested in the readings of individual sensors, but require a "net" temperature which 
might be the mean, median or any other function. Such function needs to be computed in a distributed manner. Data fusion schemes have been
used extensively in activity tracking and event detection. In fact, they improve the performance of the event detection system by handling the noisy measurements from multiple sensors \cite{fusion,fault,covmain}.

Formally, let $\vec{x^0} = (x_1^0 , \ldots , x_n^0)$ is a non-negative function representing the measurement made by the sensor nodes, and we want to compute 
$y=f( x_i^0 , \ldots , x_k^0)$ over a fixed set of nodes. Here, $f(\cdot)$ is a function that is fusion of values of nodes at a time. This function could 
mean distributed averaging \cite{conmain}, Kalman filter \cite{kalman}, Maximum likelihood estimator \cite{boyd}, linear least-squares estimator \cite{mean}, max/min, or any other decomposable operator. By fixed set of nodes, we mean that function computed locally over a small set of nodes instead on the whole graph. 
The references mentioned above show how to implement function to aggregate data in a distributed setting on a whole graph. So, it is natural to ask about the data aggregation at a local level. This is where we take advantage of our algorithm. The size of locality is controlled by $\epsilon$, and thus forming a core. As discussed in Section \ref{sec:local}, we can run any function implemented in a distributed setting in our framework using algorithm as long as the function is 
decomposable. Decomposability is a necessary condition for a function to be implemented in a distributed setting.

There are two types of data fusion schemes: decision based and value based. In decision based scheme, a group of nodes arrive at a common decision. 
For example, in case of a fire, a group of nodes can arrive at However, in value based fusion scheme, same group of sensor nodes can compute the average temperature. In our work, $f(\cdot)$ is not restricted to either decision based or value based scheme.

\begin{figure*}
\centering
\subfigure[]
{
\includegraphics[scale=0.5]{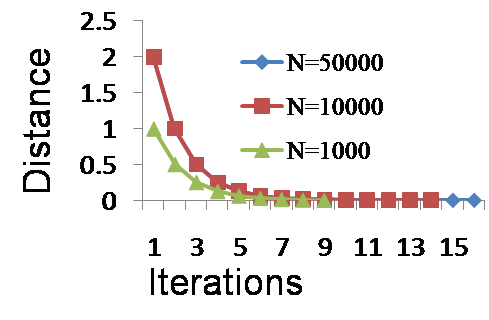}
}
\subfigure[]
{
\includegraphics[scale=0.5]{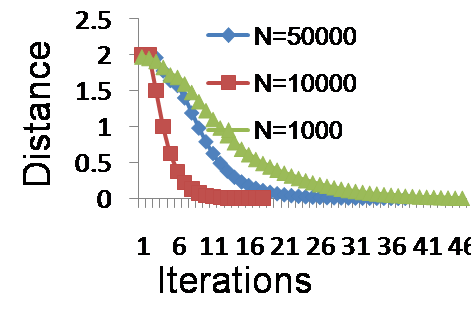}
}
\subfigure[]
{
\includegraphics[scale=0.5]{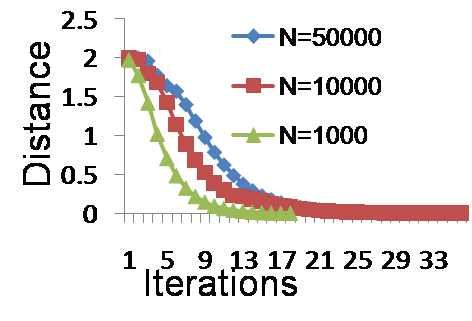}
}

\caption{The graphs are Power-law (high degree), Power-law (small degree) and random regular graph with  $\epsilon=10/n$. We can observe that the algorithm is fast and converges in few iterations.}
\label{fig:numrating}
\end{figure*}

\begin{figure*}
\centering
\subfigure[]
{
\includegraphics[scale=0.5]{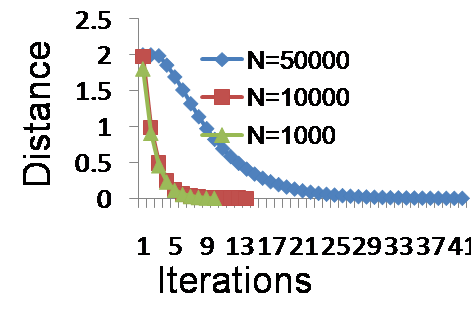}
}
\subfigure[]
{
\includegraphics[scale=0.5]{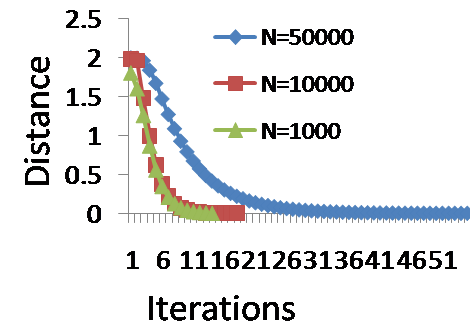}
}
\subfigure[]
{
\includegraphics[scale=0.5]{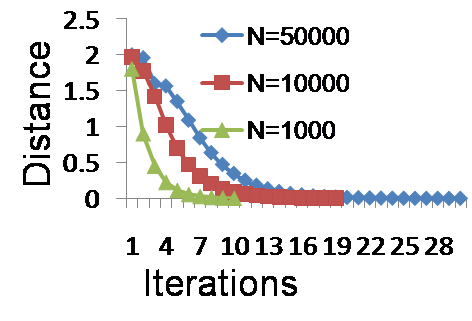}
}
\caption{The graphs are Power-law (high degree), Power-law (small degree) and random regular graph with  $\epsilon=100/n$. We can observe that the algorithm is fast and converges in few iterations.}
\label{fig:numrating}
\end{figure*}

\section{Experiments}

We experiment on synthetic data with different graphs and parameters. We choose two value of $\epsilon$: $10/n$ and $100/n$. Recall that the size of the graph needs to be more than $1/\epsilon$ for convergence. The size of the graph in our experiments is 1000 nodes, 10000 nodes and 50000 nodes. 
We generated four graphs, namely, random-regular graph where we kept the degree as 10, and two power-law graph where we start building the graph 
from low degree nodes and high degree nodes. The final graph we tested on is the cycle graph. We are particularly interested in examining the convergence rate of the algorithm. We kept the convergence criteria as $\sum_i|x^t_i-x^{t+1}_i|< 1/n$.

Fig. 3 show the result for $\epsilon=10/n$. The first chart is for power-law graph, where starting node is a high degree node. Second chart is a power-law graph with a low-degree starting node and third chart is regular-random graph with a degree 10. Figure 4 has the result $\epsilon=100/n$ and the graphs are in the same order.

In experiments, we found out that the cycle graph is slowest one and convergence takes thousands of iterations before convergence. It is not
surprising as the Markov Chains are known to converge slowly. We can relate the problem of running our algorithm on cycle graph with the random walk on a straight line.
It is well known that it takes a long time for a random walker to cross the line.

\section{Related Work}

Consensus algorithms is one of the major techniques used in data fusion \cite{conmain} and have a deep connection with Matrix Algebra \cite{matrix}, Markov chains \cite{conmain}, Distributed computing \cite{Lynch}, load balancing \cite{muthu}. Recently consensus algorithms is heavily used in sensor networks. In particular, 
many applications require a global clock synchronization,
that is all the nodes of the network need to share
a common notion of time \cite{sync1}. Therefore, clock synchronization has been an active area of research (see \cite{sync2} for a comprehensive survey).

Consensus algorithms are used in load balancing. In fact, it was one of the early methods to balance loads \cite{Cybenko}
Load balancing has shown to reduce the hot spots in the sensor networks \cite{load}. It also helps in extending the lifespan, sensor's energy might not be 
sufficient to support long communication and may require hops to forward the information \cite{load2}. The load balancing is also used frequently for routing in networks \cite{load4,load3}.

It is important to combine the data from different sensors for the purpose of event detection or other applications such as consensus \cite{fusionbook}.
Some common data fusion functions include mean distributed averaging \cite{conmain}, Kalman filter \cite{kalman}, Maximum likelihood estimator \cite{boyd} and linear least-squares estimator \cite{mean}. Nodes often collaborate for a better performance \cite{collab}. There are many applications in event detection that require data fusion \cite{collab,sensordeploy2002,fault}.

\section{Conclusion}
In this paper, we presented an algorithm that is fast, i.e., requires fewer iterations to terminate, unlike the markov chain based that converges 
asymptotically. We show that our algorithm is local and operate on a small part of graph. The algorithm fits naturally in the applications such as event detection, activity tracking, load balancing. In fact, it can complement the existing methods in data fusion as well, by making local variant by plugging the distributed fusion in Algorithm 2.

From algorithmic angle, we want to investigate the effect of {\it negative charge}, as there are fusion functions that can take negative values. 
We want to continue finding the novel applications that can take advantage of locality feature of our algorithm.

\bibliographystyle{abbrv}
\bibliography{mobicom}

\begin{thebibliography}{10}

\bibitem{article}
Article id: 57702816, http://www.informationweek.com/news.

\bibitem{collab}
S.~A. Aldosari and J.~M.~F. Moura.
\newblock Fusion in sensor networks with communication constraints.
\newblock In {\em Proceedings of the 3rd International Symposium on Information
  Processing in Sensor Networks}, IPSN '04, pages 108--115, New York, NY, USA,
  2004. ACM.

\bibitem{target1}
J.~Chen, K.~Yao, and R.~Hudson.
\newblock Source localization and beamforming.
\newblock {\em Signal Processing Magazine, IEEE}, 19(2):30--39, Mar 2002.

\bibitem{markov3}
C.-F. Chiasserini and M.~Garetto.
\newblock Modeling the performance of wireless sensor networks.
\newblock In {\em INFOCOM 2004. Twenty-third AnnualJoint Conference of the IEEE
  Computer and Communications Societies}, volume~1, pages --231, March 2004.

\bibitem{sensordeploy2002}
T.~Clouqueur, V.~Phipatanasuphorn, P.~Ramanathan, and K.~K. Saluja.
\newblock Sensor deployment strategy for target detection.
\newblock In {\em Proceedings of the 1st ACM International Workshop on Wireless
  Sensor Networks and Applications}, WSNA '02, pages 42--48, New York, NY, USA,
  2002. ACM.

\bibitem{fault}
T.~Clouqueur, K.~K. Saluja, and P.~Ramanathan.
\newblock Fault tolerance in collaborative sensor networks for target
  detection.
\newblock {\em IEEE Trans. Comput.}, 53(3):320--333, Mar. 2004.

\bibitem{Cybenko}
G.~Cybenko.
\newblock Dynamic load balancing for distributed memory multiprocessors.
\newblock {\em J. Parallel Distrib. Comput.}, 7(2):279--301, Oct. 1989.

\bibitem{load}
H.~Dai and R.~Han.
\newblock A node-centric load balancing algorithm for wireless sensor networks.
\newblock In {\em Global Telecommunications Conference, 2003. GLOBECOM '03.
  IEEE}, volume~1, pages 548--552 Vol.1, Dec 2003.

\bibitem{mean}
V.~Delouille, R.~Neelamani, and R.~Baraniuk.
\newblock Robust distributed estimation in sensor networks using the embedded
  polygons algorithm.
\newblock In {\em Information Processing in Sensor Networks, 2004. IPSN 2004.
  Third International Symposium on}, pages 405--413, April 2004.

\bibitem{load4}
J.~Gao and L.~Zhang.
\newblock Load balanced short path routing in wireless networks.
\newblock In {\em In Proc. IEEE INFOCOM 04}, pages 1099--1108, 2004.

\bibitem{muthu}
B.~Ghosh, S.~Muthukrishnan, and M.~H. Schultz.
\newblock First and second order diffusive methods for rapid, coarse,
  distributed load balancing (extended abstract).
\newblock In {\em Proceedings of the Eighth Annual ACM Symposium on Parallel
  Algorithms and Architectures}, SPAA '96, pages 72--81, New York, NY, USA,
  1996. ACM.

\bibitem{load3}
M.~Goswami, C.-C. Ni, X.~Ban, J.~Gao, X.~D. Gu, and V.~Pingali.
\newblock Load balanced short path routing in large-scale wireless networks
  using area-preserving maps.
\newblock In {\em Proceedings of the 15th ACM International Symposium on Mobile
  Ad Hoc Networking and Computing}, MobiHoc '14, pages 63--72, New York, NY,
  USA, 2014. ACM.

\bibitem{load2}
G.~Gupta and M.~Younis.
\newblock Load-balanced clustering of wireless sensor networks.
\newblock In {\em Communications, 2003. ICC '03. IEEE International Conference
  on}, volume~3, pages 1848--1852 vol.3, May 2003.

\bibitem{gupta}
P.~Gupta and P.~Kumar.
\newblock The capacity of wireless networks.
\newblock {\em Information Theory, IEEE Transactions on}, 46(2):388--404, Mar
  2000.

\bibitem{matrix}
R.~A. Horn and C.~R. Johnson, editors.
\newblock {\em Matrix Analysis}.
\newblock Cambridge University Press, New York, NY, USA, 1986.

\bibitem{kn}
A.~Jadbabaie, J.~Lin, and A.~Morse.
\newblock Coordination of groups of mobile autonomous agents using nearest
  neighbor rules.
\newblock {\em Automatic Control, IEEE Transactions on}, 48(6):988--1001, June
  2003.

\bibitem{prob}
L.~Lazos, R.~Poovendran, and J.~A. Ritcey.
\newblock Probabilistic detection of mobile targets in heterogeneous sensor
  networks.
\newblock In {\em Proceedings of the 6th International Conference on
  Information Processing in Sensor Networks}, IPSN '07, pages 519--528, New
  York, NY, USA, 2007. ACM.

\bibitem{Lynch}
N.~A. Lynch.
\newblock {\em Distributed Algorithms}.
\newblock Morgan Kaufmann Publishers Inc., San Francisco, CA, USA, 1996.

\bibitem{target2}
D.~McErlean and S.~Narayanan.
\newblock Distributed detection and tracking in sensor networks.
\newblock In {\em Signals, Systems and Computers, 2002. Conference Record of
  the Thirty-Sixth Asilomar Conference on}, volume~2, pages 1174--1178 vol.2,
  Nov 2002.

\bibitem{covmain}
S.~Meguerdichian, F.~Koushanfar, M.~Potkonjak, and M.~Srivastava.
\newblock Coverage problems in wireless ad-hoc sensor networks.
\newblock In {\em INFOCOM 2001.}, volume~3, pages 1380--1387 vol.3, 2001.

\bibitem{kalman}
R.~Olfati-Saber.
\newblock Distributed kalman filter with embedded consensus filters.
\newblock In {\em Decision and Control, 2005 and 2005 European Control
  Conference. CDC-ECC '05. 44th IEEE Conference on}, pages 8179--8184, Dec
  2005.

\bibitem{con1}
R.~Olfati-Saber.
\newblock Distributed kalman filtering for sensor networks.
\newblock In {\em Decision and Control, 2007 46th IEEE Conference on}, pages
  5492--5498, Dec 2007.

\bibitem{conmain}
R.~Olfati-Saber, J.~Fax, and R.~Murray.
\newblock Consensus and cooperation in networked multi-agent systems.
\newblock {\em Proceedings of the IEEE}, 95(1):215--233, Jan 2007.

\bibitem{sync1}
L.~Schenato and F.~Fiorentin.
\newblock Average timesynch.
\newblock {\em Automatica}, 47(9):1878--1886, Sept. 2011.

\bibitem{markov1}
R.~Shah, S.~Roy, S.~Jain, and W.~Brunette.
\newblock Data mules: modeling a three-tier architecture for sparse sensor
  networks.
\newblock In {\em Sensor Network Protocols and Applications, 2003. Proceedings
  of the First IEEE. 2003 IEEE International Workshop on}, pages 30--41, May
  2003.

\bibitem{sync2}
B.~Sundararaman, U.~Buy, and A.~D. Kshemkalyani.
\newblock Clock synchronization for wireless sensor networks: A survey.
\newblock {\em Ad Hoc Networks (Elsevier}, 3:281--323, 2005.

\bibitem{tong}
H.~Tong, C.~Faloutsos, and J.-Y. Pan.
\newblock Fast random walk with restart and its applications.
\newblock In {\em Proceedings of the Sixth International Conference on Data
  Mining}, ICDM '06, pages 613--622, Washington, DC, USA, 2006. IEEE Computer
  Society.

\bibitem{fusionbook}
P.~K. Varshney.
\newblock {\em Distributed Detection and Data Fusion}.
\newblock Springer-Verlag New York, Inc., Secaucus, NJ, USA, 1st edition, 1996.

\bibitem{boyd}
L.~Xiao, S.~Boyd, and S.~Lall.
\newblock A scheme for robust distributed sensor fusion based on average
  consensus.
\newblock In {\em Proceedings of the 4th International Symposium on Information
  Processing in Sensor Networks}, IPSN '05, Piscataway, NJ, USA, 2005. IEEE
  Press.

\bibitem{fusion}
G.~Xing, R.~Tan, B.~Liu, J.~Wang, X.~Jia, and C.-W. Yi.
\newblock Data fusion improves the coverage of wireless sensor networks.
\newblock In {\em Proceedings of the 15th Annual International Conference on
  Mobile Computing and Networking}, MobiCom '09, pages 157--168, New York, NY,
  USA, 2009. ACM.

\bibitem{markov2}
T.~Zheng, S.~Radhakrishnan, and V.~Sarangan.
\newblock Pmac: an adaptive energy-efficient mac protocol for wireless sensor
  networks.
\newblock In {\em Parallel and Distributed Processing Symposium, 2005.
  Proceedings. 19th IEEE International}, pages 8 pp.--, April 2005.

\end{thebibliography}

\end{document}